\documentclass[preprint,12pt]{elsarticle}
\usepackage[margin=1in]{geometry}

\usepackage{graphicx}
\usepackage{amssymb}

\usepackage{caption}

\usepackage{stackrel}                                          
\usepackage{verbatim} 
\usepackage{chemfig}
\usepackage{mathtools}
\usepackage{amsmath} 
\usepackage{amsthm} 

\usepackage{listings}
\usepackage{color}

\definecolor{dkgreen}{rgb}{0,0.6,0}
\definecolor{gray}{rgb}{0.5,0.5,0.5}
\definecolor{mauve}{rgb}{0.58,0,0.82}

\newtheorem{Theorem}{Theorem}[]
\newtheorem{Lemma}[Theorem]{Lemma}

\newtheorem{Corollary}[Theorem]{Corollary}

\graphicspath{{Pictures/}}
\journal{Mathematical Biosciences}

\begin{document}

\begin{frontmatter}


\title{Precluding Oscillations in Michaelis-Menten Approximations of Dual-site Phosphorylation Systems}



\author{Hwai-Ray Tung}

\address{Brown University}

\begin{abstract}
Oscillations play a major role in a number of biological systems, from predator-prey models of ecology to circadian clocks. In this paper we focus on the question of whether oscillations exist within dual-site phosphorylation systems. Previously, Wang and Sontag showed, using monotone systems theory, that the Michaelis-Menten (MM) approximation of the distributive and sequential dual-site phosphorylation system lacks oscillations. However, biological systems are generally not purely distributive; there is generally some processive behavior as well. Accordingly, this paper focuses on the MM approximation of a general sequential dual-site phosphorylation system that contains both processive and distributive components, termed the composite system. Expanding on the methods of Bozeman and Morales, we preclude oscillations in the MM approximation of the composite system. This implies the lack of oscillations in the MM approximations of the processive and distributive systems, shown previously, as well as in the MM approximation of the partially processive and partially distributive mixed-mechanism system.
\end{abstract}


\end{frontmatter}


\section{Introduction}
Oscillations appear in a number of biological systems, including predator-prey models, the lighting up of fireflies \citep{strogatz2014nonlinear}, auditory hair bundles, and cytoskeletal structures. Within the context of biochemistry, much of the study of oscillations centers around genetic oscillations, which have been shown to play a role in circadian clocks, the segmentation of vertebrates, and the activity of the tumor suppressor gene p53 \citep{oscInBio}. Oscillations also appear in dual-site phosphorylation networks, the networks of interest in this paper. 

Phosphorylation networks can modify, activate, or deactivate proteins and have been observed in membrane receptors, protein kinases, transcription factors, cell cycle regulators, and circadian clock proteins, to name a few. Dual-site phosphorylation networks modify proteins by allowing at most one phosphate group to bind onto each of the two possible sites on the protein. In allowing multiple phosphate groups to attach, a protein can exhibit multiple types of behavior \cite{phosBackgrnd}.

The importance of oscillations in phosphorylation networks can be seen from the \textbf{mitogen-activated protein kinase (MAPK) cascades}. These cell signalling networks govern, among others, gene expression, cell proliferation, cell survival and death, and cell motility \cite{chang-karin} and often include dual-site phosphorylation networks \citep{phosInMAPK}. The sustained oscillations observed in MAPK cascades may be what enables a cell to time its stages of development  \citep{yeast-mapk-oscillations,oscillations-mapk-cancer}.

In order to simultaneously examine a wide variety of dual-site phosphorylation networks, we focus on the \textbf{processive-distributive composite} network, or \textbf{composite} network:

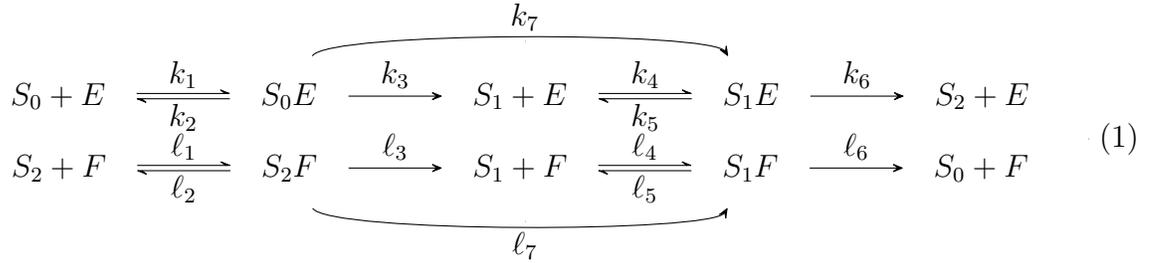
\makeatletter
\definearrow1{s>}{%
\ifx\@empty#1\@empty
\expandafter\draw\expandafter[\CF@arrow@current@style,-CF](\CF@arrow@start@node)--(\CF@arrow@end@node);%
\else
\def\curvedarrow@style{shorten <=\CF@arrow@offset,shorten >=\CF@arrow@offset,}%
\CF@expadd@tocs\curvedarrow@style\CF@arrow@current@style
\expandafter\draw\expandafter[\curvedarrow@style,-CF](\CF@arrow@start@name)..controls#1..(\CF@arrow@end@name);
\fi
}

\vspace{10pt}

\addtocounter{equation}{1}

\begin{center}
\schemestart
$S_0 + E$
 \arrow{<=>[$k_1$][$k_2$]}
$S_0E$
 \arrow{->[$k_3$]}
$S_1 + E$
 \arrow{<=>[$k_4$][$k_5$]}
$S_1E$
 \arrow{->[$k_6$]} 
 \arrow(@c2--@c4){s>[+(60:1)and+(120:1)][$k_7$]}[,2]
$S_2 + E$
\schemestop
\schemestart
$S_2 + F$
 \arrow{<=>[$\ell_1$][$\ell_2$]}
$S_2F$
 \arrow{->[$\ell_3$]}
$S_1 + F$
 \arrow{<=>[$\ell_4$][$\ell_5$]}
$S_1F$
 \arrow{->[$\ell_6$]} 
 \arrow(@c2--@c4){s>[+(-60:1)and+(-120:1)][$\ell_7$]}[,2]
$S_0 + F$
\schemestop
\chemmove{
  \filldraw[black] (-6.69,-.6) circle (0pt) node[anchor=north] {$\ell_7$};
  \filldraw[black] (-6.69,1.8) circle (0pt) node[anchor=south] {$k_7$};
  \filldraw[black] (.8,.5) circle (0pt) node[anchor=west] {$(\arabic{equation})$};
}
\end{center}

\addtocounter{equation}{-1}

\begin{equation}
    \color{white}
  \label{chemRxn}
\end{equation}

\vspace{-25pt}

\noindent This network was first explored by Suwanmajo and Krishnan under the name ``mixed composite phosphorylation model" and was shown to exhibit (under certain conditions) oscillations \cite{mixedNet}. It also encompasses three common dual-site phosphorylation networks we will later encounter.

The most studied network encompassed by the composite network is the network with a sequential and \textbf{distributive} mechanism, displayed in \eqref{distChemRxn}.

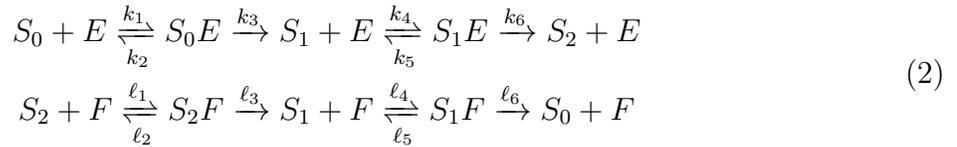
\begin{equation}
\label{distChemRxn}
\begin{gathered}
S_0 + E \xrightleftharpoons[k_2]{k_1} S_0E \xrightarrow{k_3} S_1 + E \xrightleftharpoons[k_5]{k_4} S_1E \xrightarrow{k_6} S_2 + E \\
S_2 + F \xrightleftharpoons[\ell_2]{\ell_1} S_2F \xrightarrow{\ell_3} S_1 + F \xrightleftharpoons[\ell_5]{\ell_4} S_1F \xrightarrow{\ell_6} S_0 + F
\end{gathered}
\end{equation}

\noindent Note that \eqref{distChemRxn} can be derived from \eqref{chemRxn} by setting $k_7 = \ell_7 = 0$. The mechanism is called sequential since the phosphate groups attach to the binding sites in a certain sequence; the phosphate group will always bind to site $1$ before site $2$. The mechanism is called distributive since $S_1E$ has the choice of distributing to $S_2+E$ or back to $S_1+E$, and similarly for $S_1F$. It is unknown whether network \eqref{distChemRxn} exhibits oscillations, although Wang and Sontag showed through monotone systems theory that its Michaelis-Menten (MM) approximation does not exhibit oscillations \cite{Wang2008}. The same result was shown by Bozeman and Morales using the simpler Dulac's criterion \cite{bozeMoral2016}.

Another mechanism is the \textbf{processive} mechanism. A dual-site phosphorylation network with sequential and processive mechanism can be seen in \eqref{procChemRxn}. 

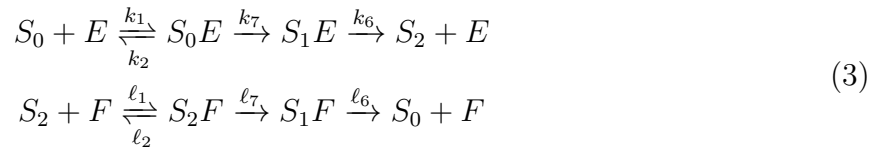
\begin{equation}
\label{procChemRxn}
\begin{gathered}
S_0 + E \xrightleftharpoons[k_2]{k_1} S_0E \xrightarrow{k_7} S_1E \xrightarrow{k_6} S_2 + E \\
S_2 + F \xrightleftharpoons[\ell_2]{\ell_1} S_2F \xrightarrow{\ell_7} S_1F \xrightarrow{\ell_6} S_0 + F
\end{gathered}
\end{equation}

\noindent We can obtain \eqref{procChemRxn} from \eqref{chemRxn} by setting rate constants $k_3=k_4=k_5 = \ell_3=\ell_4=\ell_5 = 0$. The mechanism is called processive since $S_1E$, unlike in a distributive mechanism, can only proceed to $S_2+E$, and similarly for $S_1F$. Processive systems converge to a unique equilibrium and thus do not exhibit oscillations \citep{processiveConverge, eithun2017, Rao2017}.

In reality, dual-site phosphorylation networks are generally neither purely distributive nor purely processive. Thus, this paper will focus on the existence of oscillations in the MM approximation of the composite system \eqref{chemRxn}, which we saw specializes to the distributive and processive systems \eqref{distChemRxn} and \eqref{procChemRxn} respectively. Section 2 details the process of obtaining the MM approximation. Section 3 precludes oscillations for the composite system \eqref{chemRxn}, which implies two previously shown results, the lack of oscillations in the MM approximations of the distributive system and the processive system. In addition, this precludes oscillations in the MM approximation of the \textbf{mixed-mechanism system}, a partially processive and partially distributive network seen later in \eqref{coro3}. Interestingly, the original mixed-mechanism network is able to oscillate \cite{mixedNet}. 

\section{Michaelis-Menten System and Reduction}
This work focuses on the composite network \eqref{chemRxn}, and this section will closely match the methodology of Bozeman and Morales \cite{bozeMoral2016}. For convenience, we will refer to the compounds $S_0E, S_1E, S_2F,$ and $S_1F$ as $C_1, C_2, C_3,$ and $C_4$, respectively. We also use square brackets to denote the concentration of the corresponding substance. Then, using mass action kinetics we arrive at the following system of differential equations:

\begin{subequations}
\label{unreduced}
\begin{align}
    \frac{d[S_0]}{dt} &= \ell_6[C_4] - k_1[S_0][E]+k_2[C_1], \label{eqS0}\\
    \frac{d[S_2]}{dt} &= k_6[C_2] - \ell_1[S_2][F]+\ell_2[C_3], \label{eqS2}\\
    \frac{d[S_1]}{dt} &= k_3[C_1] - k_4[S_1][E]+k_5[C_2] + \ell_3[C_3] + \ell_5[C_4] - \ell_4[S_1][F], \label{del1}\\
    \frac{d[E]}{dt} &= (k_2+k_3)[C_1] + (k_5+k_6)[C_2] -k_1[S_0][E]-k_4[S_1][E], \label{del2}\\
    \frac{d[F]}{dt} &= (\ell_2+\ell_3)[C_3] + (\ell_5+\ell_6)[C_4] -\ell_1[S_2][F]-\ell_4[S_1][F], \label{del3}\\
    \frac{d[C_1]}{dt} &= k_1[S_0][E] - (k_2+k_3)[C_1] -k_7[C_1], \label{eqC1}\\
    \frac{d[C_2]}{dt} &= k_4[S_1][E] - (k_5+k_6)[C_2] +k_7[C_1], \label{eqC2}\\
    \frac{d[C_3]}{dt} &= \ell_1[S_2][F] - (\ell_2+\ell_3)[C_3] -\ell_7[C_3], \label{eqC3}\\
    \frac{d[C_4]}{dt} &= \ell_4[S_1][F] - (\ell_5+\ell_6)[C_4] +\ell_7[C_3].\label{eqC4}
\end{align}
\end{subequations}

Further examining the composite network \eqref{chemRxn} or the differential equations \eqref{unreduced}, we find that the total amounts of substrate, enzyme $E$, and enzyme $F$, respectively, are conserved, a consequence of the conservation of matter. This gives the conservation laws 

\begin{subequations}
\label{conserv}
\begin{align}
S_T &= [S_0] + [S_1] + [S_2] + [C_1]+[C_2]+[C_3]+[C_4], \label{conserv1}\\
E_T &= [E] + [C_1]+[C_2], \label{conserv2}\\
F_T &= [F] + [C_3]+[C_4].\label{conserv3}
\end{align}
\end{subequations}

Note that we may deduce equation \eqref{del1} by taking the derivative of equation \eqref{conserv1} with respect to time and substituting in equations \eqref{eqS0}, \eqref{eqS2}, \eqref{eqC1}, \eqref{eqC2}, \eqref{eqC3}, and \eqref{eqC4}. Similarly, we may deduce \eqref{del2} from equations \eqref{conserv2}, \eqref{eqC1}, and \eqref{eqC2}, and we may deduce \eqref{del3} from equations \eqref{conserv3}, \eqref{eqC3}, and \eqref{eqC4}. As such, we replace \eqref{del1}, \eqref{del2}, and \eqref{del3} with the conservation equations \eqref{conserv}, reducing the number of differential equations to six. Now, we use Michaelis-Menten theory to reduce our system. We scale all variables other than $S_0, S_1, S_2, S_T$ by $\varepsilon$ to get

\begin{equation}
\label{scaleEpsil}
E_T = \varepsilon \widetilde{E_T}, \quad F_T = \varepsilon \widetilde{F_T}, \quad [C_i] = \varepsilon [\widetilde{C_i}], \quad [E] = \varepsilon [\widetilde{E}], \quad [F] = \varepsilon [\widetilde{F}], \quad \tau = \varepsilon t.
\end{equation}

The idea behind the scaling is the assumption that the concentration of intermediates is low relative to the substrates. Substituting in the scalings from \eqref{scaleEpsil}, we obtain

\begin{equation}
    \label{subepsilon}
    \begin{aligned}
        \frac{d[S_0]}{d\tau} &= \ell_6[\widetilde{C_4}] - k_1[S_0][\widetilde{E}]+k_2[\widetilde{C_1}], \\
        \frac{d[S_2]}{d\tau} &= k_6[\widetilde{C_2}] - \ell_1[S_2][\widetilde{F}]+\ell_2[\widetilde{C_3}], \\
        \varepsilon\frac{d[\widetilde{C_1}]}{d\tau} &= k_1[S_0][\widetilde{E}] - (k_2+k_3)[\widetilde{C_1}] -k_7[\widetilde{C_1}],\\
        \varepsilon\frac{d[\widetilde{C_2}]}{d\tau} &= k_4[S_1][\widetilde{E}] - (k_5+k_6)[\widetilde{C_2}] +k_7[\widetilde{C_1}],\\
        \varepsilon\frac{d[\widetilde{C_3}]}{d\tau} &= \ell_1[S_2][\widetilde{F}] - (\ell_2+\ell_3)[\widetilde{C_3}] -\ell_7[\widetilde{C_3}],\\
        \varepsilon\frac{d[\widetilde{C_4}]}{d\tau} &= \ell_4[S_1][\widetilde{F}] - (\ell_5+\ell_6)[\widetilde{C_4}] +\ell_7[\widetilde{C_3}],\\
    \end{aligned}
\end{equation}

\noindent with new conservation equations

\begin{equation}
    \label{subepsilonConserv}
    \begin{aligned}
        S_T &= [S_0] + [S_1] + [S_2] +  \varepsilon[\widetilde{C_1}]+ \varepsilon[\widetilde{C_2}]+ \varepsilon[\widetilde{C_3}]+ \varepsilon[\widetilde{C_4}], \\
        \widetilde{E_T} &= [\widetilde{E}] + [\widetilde{C_1}]+[\widetilde{C_2}], \\
        \widetilde{F_T} &= [\widetilde{F}] + [\widetilde{C_3}]+[\widetilde{C_4}].
    \end{aligned}
\end{equation}

Now, we approximate the system by setting $\varepsilon = 0$ in \eqref{subepsilon} and \eqref{subepsilonConserv}. This yields the new differential equations for the substrates

\begin{equation}
    \label{subepsilonZeroSubstrate}
    \begin{aligned}
        \frac{d[S_0]}{d\tau} &= \ell_6[\widetilde{C_4}] - k_1[S_0][\widetilde{E}]+k_2[\widetilde{C_1}], \\
        \frac{d[S_2]}{d\tau} &= k_6[\widetilde{C_2}] - \ell_1[S_2][\widetilde{F}]+\ell_2[\widetilde{C_3}],
    \end{aligned}
\end{equation}

\noindent the new equations for the intermediates

\begin{equation}
    \label{subepsilonZeroIntermeds}
    \begin{aligned}
        0 &= k_1[S_0][\widetilde{E}] - (k_2+k_3)[\widetilde{C_1}] -k_7[\widetilde{C_1}],\\
        0 &= k_4[S_1][\widetilde{E}] - (k_5+k_6)[\widetilde{C_2}] +k_7[\widetilde{C_1}],\\
        0 &= \ell_1[S_2][\widetilde{F}] - (\ell_2+\ell_3)[\widetilde{C_3}] -\ell_7[\widetilde{C_3}],\\
        0 &= \ell_4[S_1][\widetilde{F}] - (\ell_5+\ell_6)[\widetilde{C_4}] +\ell_7[\widetilde{C_3}],\\
    \end{aligned}
\end{equation}

\noindent and the new conservation equations

\begin{equation}
    \label{subepsilonZeroConserv}
    \begin{aligned}
        \widetilde{S_T} &= [S_0] + [S_1] + [S_2],\\
        \widetilde{E_T} &= [\widetilde{E}] + [\widetilde{C_1}]+[\widetilde{C_2}], \\
        \widetilde{F_T} &= [\widetilde{F}] + [\widetilde{C_3}]+[\widetilde{C_4}].
    \end{aligned}
\end{equation}

From \eqref{subepsilonZeroIntermeds} we obtain the equations

\begin{equation}
  \label{solveCi}
  \begin{split}
    [\widetilde{C_1}] &= \frac{k_1}{k_2+k_3+k_7} [S_0][\widetilde{E}],\\
    [\widetilde{C_2}] &= \frac{k_4[S_1][\widetilde{E}]+k_7[\widetilde{C_1}]}{k_5+k_6}\\
                      &= \frac{k_4[S_1][\widetilde{E}]}{k_5+k_6}+\frac{k_1k_7[S_0][\widetilde{E}]}{(k_5+k_6)(k_2+k_3+k_7)},
  \end{split}
\text{     }
  \begin{split}
    [\widetilde{C_3}] &= \frac{\ell_1}{\ell_2+\ell_3+\ell_7} [S_2][\widetilde{F}],\\
    [\widetilde{C_4}] &= \frac{\ell_4[S_1][\widetilde{F}]+\ell_7[\widetilde{C_3}]}{\ell_5+\ell_6}\\
                      &= \frac{\ell_4[S_1][\widetilde{F}]}{\ell_5+\ell_6}+\frac{\ell_1\ell_7[S_2][\widetilde{F}]}{(\ell_5+\ell_6)(\ell_2+\ell_3+\ell_7)}.
  \end{split}
\end{equation}

Substituting the expressions for $[\widetilde{C_i}]$ in \eqref{solveCi} into \eqref{subepsilonZeroConserv} gives

\begin{equation}
    \label{ETFT}
    \begin{aligned}
        \widetilde{E_T} &= [\widetilde{E}]\left[1 + \frac{k_1}{k_2+k_3+k_7} [S_0] + \frac{k_4[S_1]}{k_5+k_6}+\frac{k_1k_7[S_0]}{(k_5+k_6)(k_2+k_3+k_7)} \right],\\
        \widetilde{F_T} &= [\widetilde{F}]\left[1 + \frac{\ell_1}{\ell_2+\ell_3+\ell_7} [S_2] + \frac{\ell_4[S_1]}{\ell_5+\ell_6}+\frac{\ell_1\ell_7[S_2]}{(\ell_5+\ell_6)(\ell_2+\ell_3+\ell_7)} \right].
    \end{aligned}
\end{equation}

Substituting \eqref{solveCi} and \eqref{ETFT} into \eqref{subepsilonZeroSubstrate} and reparametrizing as in \eqref{newVars} yields the system

\begin{equation}
\label{reducedNearly}
\begin{aligned}
\frac{d[S_0]}{d\tau} &= \frac{a_1 [S_1] [\widetilde{F_T}]}{1+c_2[S_1]+d_1[S_2]} + \frac{a_2 [S_2] [\widetilde{F_T}]}{1+c_2[S_1]+d_1[S_2]} - \frac{a_3 [S_0] [\widetilde{E_T}]}{1+b_1[S_0]+c_1[S_1]},\\
\frac{d[S_2]}{d\tau} &= \frac{a_4 [S_1] [\widetilde{E_T}]}{1+b_1[S_0]+c_1[S_1]} + \frac{a_5 [S_0] [\widetilde{E_T}]}{1+b_1[S_0]+c_1[S_1]} - \frac{a_6 [S_2] [\widetilde{F_T}]}{1+c_2[S_1]+d_1[S_2]}, \\
\end{aligned}
\end{equation}

\noindent where

\begin{equation}
  \label{newVars}
  \begin{split}
    b_1 &= \frac{k_1(k_5+k_6+k_7)}{(k_2+k_3+k_7)(k_5+k_6)}, \\
    c_1 &= \frac{k_4}{k_5+k_6}, \\
    a_1 &= \frac{\ell_4\ell_6}{\ell_5+\ell_6}, \\
    a_2 &= \frac{\ell_1\ell_6\ell_7}{(\ell_2+\ell_3+\ell_7)(\ell_5+\ell_6)}, \\
    a_3 &= \frac{k_1(k_3+k_7)}{(k_2+k_3+k_7)}, \\
  \end{split}
\quad \quad \quad
  \begin{split}
    d_1 &= \frac{\ell_1(\ell_5+\ell_6+\ell_7)}{(\ell_2+\ell_3+\ell_7)(\ell_5+\ell_6)}, \\
    c_2 &= \frac{\ell_4}{\ell_5+\ell_6}, \\
    a_4 &= \frac{k_4k_6}{k_5+k_6}, \\
    a_5 &= \frac{k_1k_6k_7}{(k_2+k_3+k_7)(k_5+k_6)}, \\
    a_6 &= \frac{\ell_1(\ell_3+\ell_7)}{(\ell_2+\ell_3+\ell_7)},
  \end{split}
\end{equation}

\noindent with conservation equation 

\begin{equation}
\label{reducedNearlyConserv}
\widetilde{S_T} = [S_0] + [S_1] + [S_2].
\end{equation}

From \eqref{reducedNearlyConserv}, we obtain $[S_1] = \widetilde{S_T} - [S_0] - [S_2]$. Substituting this into \eqref{reducedNearly} yields a system of only two variables, as seen below. This system is our \textbf{MM approximation}.

\begin{equation}
\label{reduced}
\begin{aligned}
\frac{d[S_0]}{d\tau} &= \frac{(a_1 (\widetilde{S_T} - [S_0] - [S_2])+a_2 [S_2]) [\widetilde{F_T}]}{1+c_2(\widetilde{S_T} - [S_0] - [S_2])+d_1[S_2]} - \frac{a_3 [S_0] [\widetilde{E_T}]}{1+b_1[S_0]+c_1(\widetilde{S_T} - [S_0] - [S_2])},\\
\frac{d[S_2]}{d\tau} &= \frac{(a_4 (\widetilde{S_T} - [S_0] - [S_2]) + a_5 [S_0]) [\widetilde{E_T}]}{1+b_1[S_0]+c_1(\widetilde{S_T} - [S_0] - [S_2])} - \frac{a_6 [S_2] [\widetilde{F_T}]}{1+c_2(\widetilde{S_T} - [S_0] - [S_2])+d_1[S_2]}.
\end{aligned}
\end{equation}

\section{Precluding Oscillations in the MM Approximation}
To analyze \eqref{reduced}, we use Dulac's Criterion, which is also called the Bendixson-Dulac theorem and Bendixson's criterion. 

\begin{Lemma}(Dulac's Criterion)
\label{dulac}
Let $\frac{dx}{dt} = f(x, y)$ and $\frac{dy}{dt} = g(x, y)$ be a system of ODEs defined on a simply connected region $D$ of $\mathbb{R}^2$. If $\frac{\partial f}{\partial x} + \frac{\partial g}{\partial y} $ is always positive or always negative on the region $D$, the system does not exhibit oscillations contained in region $D$.
\end{Lemma}

We are now ready to prove the main result of this paper, Theorem \ref{mainRes1}.

\begin{Theorem}
\label{mainRes1}
Let $a_1 + a_3 + a_4 + a_6 > 0$. Then the MM approximation of the composite system \eqref{reduced} cannot oscillate.
\end{Theorem}

\begin{proof}
Following the notation in Dulac's Criterion, let $f([S_0], [S_2]) \coloneqq \frac{d[S_0]}{d\tau}$, $g([S_0], [S_2]) \coloneqq \frac{d[S_2]}{d\tau}$, and $H \coloneqq \frac{\partial f}{\partial [S_0]} + \frac{\partial g}{\partial [S_2]}$ on the region $D \coloneqq \{([S_0], [S_2]) \in \mathbb{R}^2 \quad | \quad [S_0] \geq 0, [S_2] \geq 0, [S_0] + [S_2] \leq \widetilde{S_T}\}$. Then 

\begin{equation}
\label{dfdxdgdy}
\begin{aligned}
\frac{\partial f}{\partial [S_0]} = &-E_T\frac{a_3c_1(\widetilde{S_T}-[S_2]) + a_3}{(b_1[S_0] + c_1(\widetilde{S_T} - [S_0]-[S_2]) + 1)^2} \\
&- F_T\frac{(a_1d_1-a_2c_2)[S_2] + a_1}{(d_1[S_2] + c_2(\widetilde{S_T} - [S_0] - [S_2]) +1)^2},\\
\frac{\partial g}{\partial [S_2]} = &-E_T\frac{(a_4b_1- a_5c_1)[S_0] + a_4}{(b_1[S_0] + c_1(\widetilde{S_T} - [S_0]-[S_2]) + 1)^2} \\
&- F_T\frac{a_6c_2(\widetilde{S_T} - [S_0]) + a_6}{(d_1[S_2]+ c_2(\widetilde{S_T}- [S_0] - [S_2])+1)^2}. \\
\end{aligned}
\end{equation}

\noindent From equations \eqref{newVars}, it follows that

\begin{equation}
\label{subOldVarsBack}
\begin{aligned}
a_1d_1-a_2c_2 &= \frac{\ell_1 \ell_4 \ell_6 (\ell_5+\ell_6)}{(\ell_2 + \ell_3+\ell_7)(\ell_5+\ell_6)^2} \geq 0,\\
a_4b_1- a_5c_1 &= \frac{k_1k_4k_6(k_5 + k_6)}{(k_2+k_3+k_7)(k_5+k_6)^2} \geq 0. \\
\end{aligned}
\end{equation}

\noindent From the definition of our domain $D$, we obtain the inequalities

\begin{equation}
\label{domainIneq}
\widetilde{S_T} - [S_0]\geq 0, \quad \widetilde{S_T} - [S_2] \geq 0.
\end{equation}

From the hypothesis of our theorem, $a_1 + a_3 + a_4 + a_6 > 0$, which implies that $a_1, a_3, a_4,$ or $a_6$ is positive. Without loss of generality, let $a_1 > 0$. As reaction rate constants are non-negative, so are $a_3, a_4$, and $a_6$. Applying these assumptions with inequalities \eqref{subOldVarsBack} and \eqref{domainIneq} to equations \eqref{dfdxdgdy}, we find that $\frac{\partial f}{\partial [S_0]} < 0$ and $\frac{\partial g}{\partial [S_2]} \leq 0$. Thus, $H = \frac{\partial f}{\partial [S_0]} + \frac{\partial g}{\partial [S_2]} < 0$ and by Dulac's Criterion, the MM approximation of the composite system \eqref{reduced} cannot oscillate.
\end{proof}

The assumptions stated in Theorem \ref{mainRes1} are biologically reasonable; if $a_1$ were undefined, then $\ell_6 =0$, which would result in $S_2$ not being able to become $S_0$ in the composite network \eqref{chemRxn}. Similarly, $a_3, a_4,$ and $a_6$ are defined in biologically interesting networks. If $a_1 + a_3 + a_4 + a_6 = 0$, then $a_3 = 0$ which implies $k_1 = 0$ or $k_3+k_7 = 0$. Either implication would prevent $S_0$ from becoming $S_2$. In short, a network that fails the assumptions in Theorem $\ref{mainRes1}$ cannot oscillate.

With Theorem \ref{mainRes1} we obtain the following corollaries on a purely distributive network, a purely processive network, and another subnetwork known as the mixed-mechanism network, seen in \eqref{coro3}. First, we recover the earlier result that the MM approximation of the distributive system does not exhibit oscillations \citep{Wang2008, bozeMoral2016}.

\begin{Corollary}
The MM approximation of the distributive system \eqref{distChemRxn} does not admit oscillations.
\end{Corollary}

\begin{proof} The distributive network \eqref{distChemRxn} is a subnetwork of the composite network \eqref{chemRxn} where the mass action constants $k_7$ and $\ell_7$ are set to $0$ while the others are positive. From equations \eqref{newVars}, we find that $a_1 + a_3 + a_4 + a_6 > 0$. Thus, by Theorem \ref{mainRes1} the system cannot oscillate. 
\end{proof}

Next, we show that the MM approximation of the processive system does not exhibit oscillations. This is to be expected from work by Eithun, Conradi, and Shiu, who showed that processive multisite phosphorylation networks do not exhibit oscillations \citep{processiveConverge, eithun2017}. As the original does not exhibit oscillations, it is reasonable to expect the MM approximation to not exhibit oscillations as well.

\begin{Corollary}
The MM approximation of the processive system \eqref{procChemRxn} does not admit oscillations.
\end{Corollary}

\begin{proof} The processive network \eqref{distChemRxn} is a subnetwork of the composite network \eqref{chemRxn} where the mass action constants $k_3, k_4, k_5, \ell_3, \ell_4$, and $\ell_5$ are set to $0$ while the others are positive. From equations \eqref{newVars}, we find that $a_1 + a_3 + a_4 + a_6 > 0$. Thus, by Theorem \ref{mainRes1} the system cannot oscillate.
\end{proof}

The next result concerns the following \textbf{mixed-mechanism} system. 

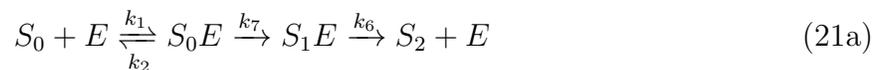
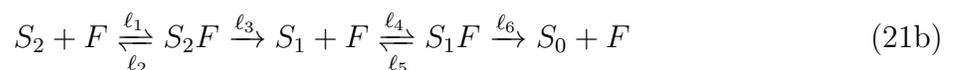
\begin{subequations}
\label{coro3}
\begin{equation}
  \label{coro3proc}
  S_0 + E \xrightleftharpoons[k_2]{k_1} S_0E \xrightarrow{k_7} S_1E \xrightarrow{k_6} S_2 + E
\end{equation}    
\begin{equation}
  \label{coro3dist}
  S_2 + F \xrightleftharpoons[\ell_2]{\ell_1} S_2F \xrightarrow{\ell_3} S_1 + F \xrightleftharpoons[\ell_5]{\ell_4} S_1F \xrightarrow{\ell_6} S_0 + F
\end{equation}
\end{subequations}

\noindent This mixed-mechanism network can be obtained from our original network \eqref{chemRxn} by setting $k_3 = k_4 = k_5 = \ell_7 = 0$. It is called the mixed-mechanism network since the phosphorylation \eqref{coro3proc} is processive while the dephosphorylation \eqref{coro3dist} is distributive.

\begin{Corollary}
\label{coro:mixed}
The MM approximation of the mixed-mechanism system \eqref{coro3} does not admit oscillations.
\end{Corollary}

\begin{proof} In \eqref{coro3}, the mass action constants $k_3, k_4, k_5$, and $\ell_7$ are set to $0$ while the others are positive. From equations \eqref{newVars}, we find that $a_1 + a_3 + a_4 + a_6 > 0$. Thus, by Theorem \ref{mainRes1} the system cannot oscillate.
\end{proof}

Corollary \ref{coro:mixed} is of particular interest because Suwanmajo and Krishnan showed that the mixed-mechanism network does exhibit oscillations \cite{mixedNet}, even though Corollary \ref{coro:mixed} states that its MM approximation does \textit{not} admit oscillations. This serves as a good reminder that the MM approximation is an approximation, and information is lost in exchange for ease of algebra. Rao also obtains Corollary \ref{coro:mixed} in independent work and uses it with the Poincare-Bendixson theorem to show that the MM approximation of the mixed-mechanism system is asymptotically stable \cite{Rao2018}.

\section{Discussion}
The goal of this paper was to investigate oscillations in the composite network, a general sequential dual-site phosphorylation network with both processive and distributive elements. The composite system was of interest for two reasons. The first was that post-translational modification networks in biochemistry are generally neither purely distributive nor purely processive. The second was that any results for the composite network subsumes results for the many subnetworks. 

In this work, we precluded oscillations in the MM approximation of the composite system. Using this result, we confirmed the lack of oscillations in the MM approximation of the distributive system and processive system and also showed that the MM approximation of the mixed-mechanism system exhibits no oscillations, contrary to the original mixed-mechanism system \cite{mixedNet}. 

The natural next question is to increase the amount of phosphorylation sites and attempt to prove or disprove the existence of oscillations. Naturally, more phosphorylation sites results in more networks. Instead of having to work each network out, analyzing a composite network that contains all $n$-site networks may yield a result that subsumes results for all subnetworks. If attempting to preclude oscillations, tools other than Dulac's Criterion, which is applicable only for $2-$dimensional systems, will be needed. Promising tools include monotone systems theory and Lyapunov functions. Monotone systems theory was employed by Wang and Sontag to preclude oscillations in the MM approximation of the dual-site distributive system, and could be used on $n$-site networks as well \cite{Wang2008}. Rao found success using Lyapunov functions to preclude oscillations in special cases of the MM approximation of the $n$-site mixed-mechanism system \cite{Rao2018}.

Proving or precluding the existence of oscillations in the MM approximation of $n$-site networks will affect efforts to do so in the original. Since features of the MM approximation can often be carried up to the original (see Hell and Rendall \cite{HELL2015175}), to rule out oscillations in the MM approximation would also rule out a method to prove the existence of oscillations in the original. Should the existence of oscillations for the MM approximation be proven instead, we would likely be able to show the existence of oscillations in the original.

\section*{Acknowledgements}
This project would not have been possible without the mentorship of Dr. Anne Shiu and the support of Jonathan Tyler. The author also thanks the two anonymous referees whose detailed feedback have improved this work. The author is particularly grateful to one of the reviewers who supplied a proof that strengthened Theorem \ref{mainRes1}. This research was conducted with funding from the NSF (DMS-1460766) at the REU program of Texas A\&M University, Summer 2017.


\bibliographystyle{unsrtnat}
\bibliography{biblio.bib}

\end{document}